\documentclass[twocolumn,pra,letterpaper]{revtex4}
\usepackage{hyperref}
\usepackage{latexsym,amsmath,amssymb,amsfonts,graphicx,color,amsthm}
\usepackage{enumerate}

\newcommand{\cA}{\mathcal{A}}
\newcommand{\cB}{\mathcal{B}}

\newcommand{\cD}{\mathcal{D}}

\newcommand{\cF}{\mathcal{F}}

\newcommand{\cH}{\mathcal{H}}

\newcommand{\cL}{\mathcal{L}}

\newcommand{\cP}{\mathcal{P}}

\newcommand{\Id}{\mathbb{I}}
\newcommand{\tr}{\text{tr}}

\newtheorem{theorem}{Theorem}
\newtheorem{lemma}[theorem]{Lemma}

\begin{document}

\title{A disturbance tradeoff principle for quantum measurements}
\author{Prabha Mandayam}
\affiliation{Department of Physics, Indian Institute of Technology  Madras, Chennai   - 600036, India. }
\author{M. D. Srinivas}
\affiliation{Centre for Policy Studies, Mylapore, Chennai - 600004, India.}
\date{\today}

\begin{abstract}
We demonstrate a fundamental principle of disturbance tradeoff for quantum measurements, along the lines of the celebrated uncertainty principle: The disturbances associated with measurements performed on {\it distinct} yet {\it identically prepared} ensembles of systems in a pure state cannot {\it all} be made arbitrarily small. Indeed, we show that the average of the disturbances associated with a set of projective measurements is strictly greater than zero whenever the associated observables do not have a common eigenvector. For such measurements, we show an equivalence between disturbance tradeoff measured in terms of fidelity and the entropic uncertainty tradeoff formulated in terms of the Tsallis entropy ($T_{2}$). We also investigate the disturbances associated with the class of non-projective measurements, where the difference between the disturbance tradeoff and the uncertainty tradeoff manifests quite clearly.


\end{abstract}

\maketitle

The uncertainty principle, which is one of the cornerstones of quantum theory, has had a long history. In its original formulation by Heisenberg for canonically conjugate variables~\cite{heisenberg:ur}, the uncertainty principle was stated as an effect of the {\it disturbance} caused due to a measurement of one observable on a succeeding measurement of another. However, the subsequent mathematical formulation due to Robertson~\cite{robertson:uncertainty} and Schr\"odinger~\cite{schrodinger} in terms of variances departed from this original interpretation. Rather, they obtained a non-trivial lower bound on the product of the variances associated with the measurement of a pair of incompatible observables, performed on {\it distinct} yet {\it identically prepared} copies of a given system. The Robertson-Schr\"odinger inequality thus expresses a fundamental limitation as regards the preparation of an ensemble of systems in identical states, and is therefore a manifestation of the so-called {\it preparation uncertainty}~\cite{BLW13_Heisenberg}.

Along the same lines, the more recent entropic formulations of the uncertainty principle~\cite{ww:urSurvey} also demonstrate the existence of a fundamental tradeoff for the uncertainties associated with {\it independent} measurements of incompatible observables on {\it identically prepared} ensemble of systems. Entropic uncertainty relations (EURs) have been obtained for specific classes of observables for both the Shannon and R\'enyi entropies~\cite{bm:uncertainty, deutsch:ur, kraus:ur, maassen:ur, ivanovic:ur, sanchez:old, sanchez:entropyD2, bia:uncertaintyRenyi, ww:cliffordUR}, as well as for the Tsallis entropies~\cite{tsallis:entropy, tsallis1, rastegin2013_tsallis}.

Here, we prove the existence of a similar principle of tradeoff for the {\it disturbances} associated with the measurements of a set of observables. It is a fundamental feature of quantum theory that when an observable is measured on an ensemble of systems, the density operator of the resulting ensemble is in general different from that prior to the measurement. The distance between these two density operators is indeed a measure of the {\it disturbance} due to measurement. Different measures of distance between density operators~\cite{NCbook} give rise to different measures of disturbance. We are thus lead to a class of disturbance measures which have been used recently in the context of quantifying incompatibility of a pair of observables~\cite{PM_MDS14}. In terms of this class of disturbance measures, we demonstrate the existence of a fundamental tradeoff principle for the disturbances associated with quantum measurements performed on {\it distinct} yet {\it identically prepared} copies of a pure state.

For the class of projective measurements considered in standard quantum theory, we show that the average of the disturbances associated with a set of such measurements is strictly greater than zero whenever the associated set of observables do not have a common eigenvector. In the particular case when the disturbance is characterized by the square of the fidelity function, we show a mathematical equivalence between the disturbance due to the measurement of an observable on a pure state and the uncertainty as quantified by the Tsallis entropy ($T_{2}$) of order $2$ of the probability distribution over the outcomes of such a measurement. Our work thus provides a new operational significance to the $T_{2}$ entropy in the context of quantum information theory. We make use of some of the known results on EURs to obtain disturbance tradeoff relations for specific classes of observables. We also prove an optimal disturbance tradeoff relation for a pair of qubit observables, which is based on a new, tight $T_{2}$ EUR. 

When we consider the general class of observables given by positive operator valued measures (POVMs), the associated measurements are characterized by completely positive (CP) instruments. For this general class of non-projective measurements we show that the disturbance and uncertainty tradeoffs are significantly different; they seem to capture different aspects of the mutual incompatibility of a set of measurements.

It may be noted that uncertainty relations have also been studied in the successive measurement scenario, both in the form of entropic relations~\cite{MDS01, MDS03_successiveEUR}, and in the form of {\it error-disturbance} relations~\cite{ozawa2004, BLW13_Heisenberg,BHOW_PRL14,ColesFurrer_13} that are in line with Heisenberg's original interpretation of the uncertainty principle. In contrast, here we look at the disturbances associated with {\it distinct} measurements of incompatible observables on identically prepared ensembles of systems. Our work thus brings to light a completely novel aspect of measurement-induced-disturbance: quantum theory places a fundamental constraint on these disturbances even when the corresponding measurements are made on distinct copies of a state.

The rest of the paper is organized as follows. We define the class of disturbance measures in Sec.~\ref{sec:dist_measures} and derive the tradeoff principle for projective measurements in Sec.~\ref{sec:disturbance_tradeoff}. We discuss the equivalence between the fidelity-based measure and the $T_{2}$ entropy in Sec.~\ref{sec:DF_EUR}, using which we obtain fidelity-based disturbance tradeoff relations for specific classes of observables in Sec.~\ref{sec:dist_example}. In Sec.~\ref{sec:qubit_disturbance} we prove an optimal tradeoff relation for a pair of qubit observables. Finally, we discuss the disturbance tradeoff principle for non-projective measurements in Sec.~\ref{sec:dist_POVM}.

\section{Disturbance measures for quantum measurements}\label{sec:dist_measures}

We begin with a brief review of the mathematical formalism of quantum measurements with a discrete set of outcomes. In standard quantum theory, they correspond to observables which are self-adjoint operators with purely discrete spectra. Any such observable $A$ has a spectral resolution $A = \sum_{i}a_{i}P^{A}_{i}$, where $\{P^{A}_{i}\}$ are orthogonal projectors. Measurement of such an observable $A$ transforms the state $\rho$, as per the von Neumann-L\"uders collapse postulate, to $\Phi^{A}(\rho) = \sum_{i}P^{A}_{i}\rho P^{A}_{i}$.

More generally, a quantum observable $\cA$ with discrete outcomes is described by a collection of positive operators $\{ 0 \leq A_{i} \leq \Id\}$ that satisfy $\sum_{i}A_{i} = \Id$. The probability of obtaining outcome $i$ when measuring observable $A$ in state $\rho$ is given by $\tr[\rho A_{i}]$. For such general observables given by positive operator valued (POV) measures $\cA\sim\{A_{i}\}$, there is no canonical specification of the post-measurement state; the associated measurement transformation can now be chosen as any CP {\it instrument} $\Phi^{\cA}$ {\it implementing} the POVM $\cA$~\cite{Heinosaari_book}.

A CP instrument $\Phi^{\cA}$ is a collection of completely positive linear maps $\Phi^{\cA}_{i}$ such that the probability of realizing outcome $i$ is given by $\tr[\Phi_{i}^{\cA}(\rho)] = \tr[\rho A_{i}]$, for all states $\rho$. The overall transformation of state $\rho$ by instrument $\Phi^{\cA}$ is described by a {\it quantum channel}, that is, a completely positive trace-preserving (CPTP) map (also denoted by $\Phi^{\cA}$):
\[\Phi^{\cA}(\rho) = \sum_{i}\Phi^{\cA}_{i}(\rho).\]
It is well known that any CPTP channel can be represented in the form $\Phi^{\cA}(\rho) = \sum_{i}K_{i}\rho K_{i}^{\dagger}$, where the {\it Kraus operators} $\{K_{i}\}$ satisfy $\sum_{i}K_{i}^{\dagger}K_{i} = \Id$.

The same observable can indeed be implemented by several different instruments. One simple implementation of a measurement of observable $\cA\sim\{A_{i}\}$ is given by the {\it L\"uders instrument} $\Phi_{\cL}^{\cA}$, in which the post-measurement state after a measurement of observable $\cA$ on state $\rho$ is given by
\[\Phi_{\cL}^{\cA}(\rho) = \sum_{i}A_{i}^{1/2}\rho A_{i}^{1/2}.\]



For a general quantum measurement on an ensemble of systems in state $\rho$, the post-measurement state $\Phi^{\cA}(\rho)$ of the ensemble is thus described via the action of a CPTP map $\Phi^{\cA}$, often called the {\it measurement channel}. The {\it distance} between the states $\rho$ and $\Phi^{\cA}(\rho)$ is a valid measure of the {\it disturbance} caused to state $\rho$ by a measurement of $A$. The disturbance due to the measurement $\cA$ on $\rho$ can therefore be estimated by any of the following measures~\cite{NCbook, PM_MDS14}: the trace-distance $D_{1}\left( \Phi^{\cA}(\rho) , \rho\right)$, the fidelity-based distance $D_{F}(\Phi^{\cA}(\rho), \rho)$, or the operator norm $\parallel \Phi^{\cA}(\rho) - \rho \parallel$. We formally define the corresponding disturbance measures here.
\begin{eqnarray}
\mathfrak{D}_{1}(\cA; \rho) &\equiv& \frac{1}{2}\tr \left\vert \Phi^{\cA}(\rho) - \rho \right\vert \nonumber \\
\mathfrak{D}_{F}(\cA;\rho) &\equiv& 1 - F^{2}[\Phi^{\cA}(\rho), \rho] \nonumber \\
\mathfrak{D}_{\infty}(\cA;\rho) &\equiv& \parallel \Phi^{\cA}(\rho) - \rho \parallel . \label{eq:disturbance}
\end{eqnarray}
Note that all three disturbance measures satisfy $0 \leq \mathfrak{D}_{\alpha}(\cA;\rho) \leq 1$, $\alpha \in \{1,F,\infty\}$, and that $\mathfrak{D}_{\alpha}(\cA;\rho) = 0$ iff $\Phi^{\cA}(\rho) = \rho$.


In the following Lemma, we state, for different classes of measurements, the necessary and sufficient conditions for a pure state to be left {\it undisturbed} by a given measurement.

\begin{lemma}[Zero-disturbance conditions for pure states] \label{lem:zero_dist}
\begin{itemize}
\item[(a)] For a projective measurement associated with a self-adjoint operator $A$ with a purely discrete spectrum, $\mathfrak{D}_{\alpha}(A;|\psi\rangle) = 0 \; (\alpha \in \{1,F,\infty\})$ if and only if $|\psi\rangle$ is an eigenstate of $A$.
\item[(b)] If $\cA\sim\{A_{i}\}$ is a POVM implemented by the L\"uders channel $\Phi_{\cL}^{\cA}$, $\mathfrak{D}_{\alpha}(\cA;|\psi\rangle) = 0 \; (\alpha \in \{1,F,\infty\})$ if and only if $|\psi\rangle$ is a common eigenstate of the operators $\{A_{i}\}$.
\item[(c)] If $\cA$ is a POVM implemented by a general CP instrument $\Phi^{\cA}(\rho)=\sum_{i}K_{i}\rho K_{i}^{\dagger}$, $\mathfrak{D}_{\alpha}(\cA;|\psi\rangle) = 0 \; (\alpha \in \{1,F,\infty\})$ if and only if the state $|\psi\rangle$ satisfies
\[ \sum_{i}|\langle\psi|K_{i}|\psi\rangle|^{2} = 1  .\]
\end{itemize}
\end{lemma}

\begin{proof}
Let $A$ be a projective measurement. By definition, $D_{\alpha}(\rho, \sigma) = 0$ if and only if $\rho = \sigma$. Therefore, $\mathfrak{D}_{\alpha}(A;|\psi\rangle) = 0$ if and only if $\Phi^{A}(|\psi\rangle\langle\psi|) = |\psi\rangle\langle\psi|$. This is equivalent to $\sum_{i}\langle (\psi|P^{A}_{i}|\psi\rangle)^{2} = 1$, which proves the part (a) of the Lemma.

For a POVM $\cA \sim \{A_{i}\}$ which is implemented by a L\"uders instrument $\Phi_{\cL}^{\cA}$, the disturbance vanishes if and only if, \[\sum_{i}A_{i}^{1/2}|\psi\rangle\langle\psi|A_{i}^{1/2} = |\psi\rangle\langle\psi|.\]
This is equivalent to the condition that $[A_{i},|\psi\rangle\langle\psi|] = 0$ for all $A_{i}$~\cite{busch_singh98, AGG02}. 
Therefore, $\mathfrak{D}_{\alpha}(\cA;|\psi\rangle) = 0$ iff the state $|\psi\rangle$ is a common eigenstate of all the $A_{i}$'s.

Finally, for a POVM $\cA$ implemented by a general CP instrument $\Phi^{\cA}$, the disturbance $D_{\alpha}(\cA;|\psi\rangle) = 0$ iff
\begin{eqnarray}
\sum_{i}K_{i}|\psi\rangle\langle\psi|K_{i}^{\dagger} &=& |\psi\rangle\langle\psi| \nonumber \\
\Leftrightarrow \sum_{i}|\langle\psi|K_{i}|\psi\rangle|^{2} &=& 1.
\end{eqnarray}
\end{proof}


\section{Disturbance Tradeoff Principle for projective measurements} \label{sec:disturbance_tradeoff}
In this and the following section we restrict our attention to projective measurements that are associated with self-adjoint operators with purely discrete spectra. We first note the following property of the average disturbance of a pair of observables.
\begin{lemma}\label{lem:avg_dist}
For a pair of observables $A$ and $B$ with purely discrete spectra, define the quantity
\begin{equation}
d_{\alpha}(A,B) \equiv \inf_{|\psi\rangle} \frac{1}{2}\left[\mathfrak{D}_{\alpha}(A;|\psi\rangle) + \mathfrak{D}_{\alpha}(B;|\psi\rangle) \right]. \label{eq:inf_disturbance}
\end{equation}
Then, $0\leq d_{\alpha}(A,B) \leq 1 \; (\alpha \in \{1,F,\infty\})$, with $d_{\alpha}(A,B) = 0$ if and only if $A$ and $B$ have a common eigenvector.
\end{lemma}
\begin{proof}
The first part simply follows from the fact that $0 \leq \mathfrak{D}_{\alpha}(X;|\psi\rangle) \leq 1$ for $X=A,B$, for $\alpha \in \{1,F,\infty\}$. Next, we know from Lemma~\ref{lem:zero_dist} that the individual disturbances vanish only for eigenstates of the corresponding observables. Therefore the average vanishes if and only if there exists a common eigenvector of $A$ and $B$.
\end{proof}

\noindent This implies the following disturbance tradeoff principle: \\
{\centering \it For any two observables $A$ and $B$ with purely discrete spectra which do not have any common eignevector, there exists a quantity $d_{\alpha}(A,B) > 0$, such that for any pure state $|\psi\rangle$, the average of the disturbances due to measurements of $A$ and $B$ (performed independently, on identically prepared copies of $|\psi\rangle$) is greater than or equal to $d_{\alpha}(A,B)$.}

In other words, the disturbances due to measurements of incompatible observables $A$ and $B$ on any pure state cannot {\it both} be made arbitrarily small; if one is {\it small}, the other must necessarily be {\it of the order of} $d_{\alpha}(A,B)$, even though the measurements are performed independently, on identically prepared copies of any given pure state. Mathematically, this may be stated as,
\begin{equation}
 \frac{1}{2}\left[\mathfrak{D}_{\alpha}(A;|\psi\rangle) + \mathfrak{D}_{\alpha}(B;|\psi\rangle)\right] \geq d_{\alpha}(A,B) > 0, \label{eq:dist_tradeoff}
\end{equation}
for all pure states $|\psi\rangle$ and observables $A$ and $B$ which do not have any common eigenvector. 

More generally, a disturbance relation for a set of observables $\{A_{1}, A_{2}, \ldots, A_{N}\}$ is a state-independent lower bound of the form
\[ \frac{1}{N} \sum_{i=1}^{N}\mathfrak{D}_{\alpha}(A_{i};|\psi\rangle) \geq d_{\alpha}(A_{1},\ldots,A_{N}), \; \forall |\psi\rangle. \]
A simple extension of Lemma~\ref{lem:avg_dist} proves that $d_{\alpha}(A_{1}, \ldots, A_{N}) > 0$ for a set of observables that do not have any common eigenvector.

It should be noted that the above disturbance tradeoff principle holds {\it only} for pure state ensembles. If we take into consideration mixed states as well, then we have for instance the maximally mixed state $\frac{\Id}{d}$ in $d$-dimensions, which is {\it not disturbed} by the measurement of any observable! This would be so whatever be the measure of disturbance that one employs. 

Finally, we note that though we have formulated the tradeoff principle using a specific class of distance measures $D_{\alpha}$, such a tradeoff principle holds for {\it any} disturbance measure which is based on a distance $D(\rho, \sigma)$ satisfying $D(\rho,\sigma) = 0$ iff $\rho = \sigma$.


The disturbance tradeoff principle demonstrated above for projective measurements seems to bear a close resemblance to the well known uncertainty tradeoff principle, especially since the lower bounds in both cases vanish iff the set of observables under consideration have a common eigenvector. Here, it should first be noted that the disturbance associated with a measurement as defined in Eq.~\eqref{eq:disturbance} does not involve the probabilities for obtaining different outcomes in the measurement. On the other hand, the various entropies which are used to quantify uncertainty measure the spread in the probability distribution over the outcomes. Thus, there is no obvious relation between the disturbance caused by a measurement and the uncertainty or the ``spread'' in the probability distribution over its outcomes.

In the case of projective measurements, it so happens that the eigenstates of the osbervable are the states which are left undisturbed by the measurement, and, they are also the states in which the spread of the probability distribution is zero. In other words, for projective measurements, the pure states with zero disturbance are also the states of zero uncertainty. We shall later see that there is no such relation between the disturbance and uncertainty for general, non-projective measurements.

For the case of projective measurements, we now show that there is in fact a mathematical equivalence between a specific disturbance measure $\mathfrak{D}_{F}(A;|\psi\rangle)$ and the Tsallis entropy $T_{2}(A;|\psi\rangle)$ which is one of the standard measures of uncertainty.

\subsection{Fidelity-based Disturbance and $T_{2}$ Entropic Uncertainty} \label{sec:DF_EUR}

 It is easy to see that there is a mathematical relation between the disturbance caused by a measurement of observable $A = \sum_{i}a_{i}P^{A}_{i}$ on state $|\psi\rangle$, as quantified by the fidelity-based measure $\mathfrak{D}_{F}$, and the spread of the probability distribution over the outcomes of a measurement of $A$ on $|\psi\rangle$. From Eq.~\eqref{eq:disturbance}, we have,
\begin{eqnarray}
\mathfrak{D}_{F}(A;|\psi\rangle)  &=&  1- F^{2}[\Phi^{A}(|\psi\rangle\langle\psi|), |\psi\rangle\langle\psi|] \nonumber \\
&=& 1 - \sum_{i}\left(p_{|\psi\rangle}^{A}(i)\right)^{2}, \label{eq:disturbance_exp}
\end{eqnarray}
where $p_{|\psi\rangle}^{A}(i) \equiv \langle\psi|P^{A}_{i}|\psi\rangle$ is the probability of obtaining outcome $i$.

Recall that the Tsallis entropy~\cite{tsallis:entropy} $T_{\beta}(\{p(i)\})$ of a discrete probability distribution $\{p(i)\}$, for any real $\beta > 0$ ($\beta \neq 1$), is defined as:
\begin{equation}
T_{\beta} (\{p(i)\})  = \frac{1}{1-\beta}\left(\sum_{i} p(i)^{\beta} - 1\right) . \label{eq:tsallis}
\end{equation}
In the limit $\beta\rightarrow 1$ the Tsallis entropy $T_{1}$ is the same as the Shannon entropy: $T_{1}(\{p(i)\}) = - \sum_{i} p(i)\log p(i)$. For $d$-dimensional distributions the Tsallis entropies satisfy,
\[0 \leq T_{\beta}(\{p(i)\}) \leq \frac{d^{1-\beta} - 1}{1 - \beta}.  \]
In particular,
\[ 0 \leq T_{2}(\{p_{i}\}) \leq 1 - \frac{1}{d}, \]
where the lower bound is attained for $p_{i} = \delta_{ik}$ for some $k$, and, the upper bound,  when $p_{i} = \frac{1}{d}, \; \forall i$.

The $T_{\beta}$ entropy defined in Eq.~\eqref{eq:tsallis} is a non-extensive entropy originally developed by Tsallis in the context of statistical physics~\cite{tsallis:entropy}. The Tsallis entropies are concave for the entire range of $\beta \in (0,\infty)$~\cite{bengtsson_2006}. Similar to the Shannon and R\'enyi entropies, the Tsallis entropies $T_{\beta}(A;|\psi\rangle)$ of the probability distribution $p_{\psi}^{A}(i)$ over the outcomes of a measurement of an observable $A$ have been widely used to quantify the {\it uncertainty} associated with the outcome of a measurement of $A$ in state $|\psi\rangle$. For instance, the uncertainty principle for canonically conjugate observables has been reformulated in terms of the Tsallis entropies~\cite{tsallis1}. More recently, Tsallis entropic uncertainty relations have been obtained for specific classes of observables in finite dimensions~\cite{majernik_2001, NKG12_stabilizerEUR, rastegin2013_tsallis, rastegin2012_T1/2}.

It is easy to see from Eq.~\eqref{eq:disturbance_exp}, that the disturbance measure $\mathfrak{D}_{F}(A;|\psi\rangle)$ is indeed the same as the $T_{2}$-entropy of the probability distribution of the outcomes of a measurement of observable $A$ on state $|\psi\rangle$:
\begin{equation}
\mathfrak{D}_{F}(A;|\psi\rangle) = 1 - \sum_{i}(p_{|\psi\rangle}^{A}(i))^{2} = T_{2}(A;|\psi\rangle). \label{eq:DF_T2equiv}
\end{equation}
Both the disturbance $\mathfrak{D}_{F}(A;|\psi\rangle)$ and the uncertainty $T_{2}(A;|\psi\rangle)$ vanish when $|\psi\rangle$ is an eigenstate of $A$. For a non-degenerate observable $A$ in $d$-dimensions, both the disturbance and the uncertainty attain the maximum value $\left(1 - \frac{1}{d}\right)$ when $|\psi\rangle$ is mutually unbiased with respect to the eigenstates of $A$.

It may however be noted that this interesting equivalence between the fidelity-based measure of disturbance and the uncertainty measure given by the $T_{2}$ Tsallis entropy holds {\it only} for pure states. For mixed states, it follows from Eqs.~\eqref{eq:disturbance} and~\eqref{eq:disturbance_exp} that the disturbance is in general less than the corresponding $T_{2}$ entropy.

\subsection{Disturbance Tradeoffs for specific classes of observables}\label{sec:dist_example}

Using the equivalence in Eq.~\eqref{eq:DF_T2equiv}, we can directly obtain disturbance tradeoff inequalities for those classes of observables for which a $T_{2}$ EUR can be obtained. In particular, here we show disturbance tradeoff relations for (a) a set of $N$ mutually unbiased bases (MUBs), and, (b) a set of dichotomic anti-commuting observables.

(a) Let $\cB_{m} \equiv \{|i_{m}\rangle, i = 1,\ldots,d\}$ $(m = 1, \ldots, N)$ denote a set of $N$ MUBs in $d$-dimensions. Recall that two bases $\cB_{m}, \cB_{n}$ are said to be mutually unbiased if their respective basis vectors satisfy,
\[ |\langle i_{m} | j_{n}\rangle|^{2}  = \frac{1}{d}, \; \forall i,j .\]
Let $p_{|\psi\rangle}^{\cB_{m}}(i) = |\langle i_{m}|\psi\rangle|^{2}$ denote the probability of obtaining the $i$th outcome while measuring $\cB_{m}$ on a pure state $|\psi\rangle$. Then, Wu {\it et al.} show that~\cite{wu:MUB_bound},
\begin{equation}
\sum_{i=1}^{d}\sum_{m=1}^{N} \left(p_{|\psi\rangle}^{\cB_{m}}(i)\right)^{2} \leq 1 + \frac{N-1}{d} . \label{eq:prob_Nd}
\end{equation}
This immediately implies the following disturbance tradeoff relation for a set of $N$ MUBs in $d$-dimensions:
\begin{equation}
 \frac{1}{N}\sum_{m=1}^{N}\mathfrak{D}_{F}(\cB_{m};|\psi\rangle) \geq \left(1 - \frac{1}{N}\right)\left(1 - \frac{1}{d}\right) . \label{eq:dist_Nd}
\end{equation}
In dimensions where a complete set of $(d+1)$ MUBs exist, it is known that Eq.~\eqref{eq:prob_Nd} is in fact an exact equality for the full set of $(d+1)$ MUBs~\cite{larsen:entropy, ivanovic:ur}. Correspondingly, the disturbance tradeoff relation in Eq.~\eqref{eq:dist_Nd} becomes an exact equality for a complete set of $(d+1)$ MUBs.

(b) Let $\{A_{1}, A_{2}, \ldots, A_{N}\}$ be a set of pairwise anticommuting operators with eigenvalues $\{\pm 1\}$:
\begin{equation}
 \{A_{j}, A_{k}\} = 0\; \forall j\neq k, \; (A_{i})^{2} = \Id \; \forall i=1,\ldots,N . \label{eq:anticommuting}
\end{equation}
Let $P_{i}^{+}$ and $P_{i}^{-}$ denote the projectors onto the positive and negative eigenspaces respectively, of observable $A_{i}$. Correspondingly, $p_{|\psi\rangle}^{A_{i}}(\pm) = \langle \psi|P^{\pm}_{i}|\psi\rangle$ are the probabilities of obtaining values $\pm 1$ while measuring $A_{i}$ in state $|\psi\rangle$. 
We obtain a disturbance tradeoff relation for such a set of dichotomic, anticommuting observables. We merely state the relation here and refer to the appendix (Sec.~\ref{sec:proof}) for the proof.
\begin{theorem}\label{thm:anticommuting}
For a set of $N$ anticommuting observables $\{A_{1}, A_{2}, \ldots, A_{N}\}$ defined in Eq.~\eqref{eq:anticommuting},
\begin{equation}
 \frac{1}{N}\sum_{i=1}^{N}\mathfrak{D}_{F}(A_{i};|\psi\rangle) \geq \frac{1}{2}\left(1 - \frac{1}{N}\right). \label{eq:anticommuting_dist}
\end{equation}
\end{theorem}

As a particular instance of Eq.~\eqref{eq:dist_Nd} and Eq.~\eqref{eq:anticommuting_dist}, we have the elegant tradeoff relation for the disturbances associated with measurements of the spin components $\sigma_{X}$, $\sigma_{Y}$ and $\sigma_{Z}$ in $d=2$:
\begin{equation}
\mathfrak{D}_{F}(\sigma_{X};|\psi\rangle) + \mathfrak{D}_{F}(\sigma_{Y}|\psi\rangle) + \mathfrak{D}_{F}(\sigma_{Z}; |\psi\rangle) = 1 . \label{pauli_dist}
\end{equation}
The corresponding $T_{2}$ EUR has been derived in~\cite{rastegin2013_tsallis}.

Finally, we note that the disturbance tradeoff relations for the fidelity-based disturbance measure $\mathfrak{D}_{F}$ imply similar disturbance tradeoff relations for the $\mathfrak{D}_{1}$ measure. This follows from the fact that for pure states, $\mathfrak{D}_{1}(A;|\psi\rangle) \geq \mathfrak{D}_{F}(A;|\psi\rangle)$~\cite{NCbook}, so that the lower bound \eqref{eq:dist_tradeoff} for the average fidelity-based disturbance is also a lower bound on the average $\mathfrak{D}_{1}$ disturbance.

\section{Optimal Disturbance Tradeoff Relation for a pair of Qubit Observables}\label{sec:qubit_disturbance}

We demonstrate the following {\it optimal} disturbance tradeoff relation for any pair of observables in a two-dimensional Hilbert space.

\begin{theorem}\label{thm:optimal_qubit}
For a pair of qubit observables $A, B$ with discrete spectra $A = \sum_{i=1}^{2}a_{i}|a_{i}\rangle\langle a_{i}|$, $B = \sum_{j=1}^{2}b_{j}|b_{j}\rangle\langle b_{j}|$, and any pure state $|\psi\rangle \in \mathbb{C}^{2}$,
\begin{equation}
\frac{1}{2}\left[\mathfrak{D}_{F}(A;|\psi\rangle) + \mathfrak{D}_{F}(B;|\psi\rangle)\right] \geq \frac{1}{2}(1-c^{2}),  \label{eq:optimal_qubit}
\end{equation}
where $c \equiv \max_{i,j=1,2}|\langle a_{i}|b_{j}\rangle|$.
\end{theorem}

The problem of finding the lower bound on the average disturbance simplifies considerably once we use the Bloch sphere representation for qubit observables. In other words, we parameterize $A$ and $B$ in terms of unit vectors $\vec{a}, \vec{b} \in \mathbb{R}^{3}$ and real parameters $\{\alpha_{i}, \beta_{i}\}$ as follows: $A = \alpha_{1}\Id + \alpha_{2}\vec{a}.\vec{\sigma}$ and $B = \beta_{1}\Id + \beta_{2}\vec{b}.\vec{\sigma}$. The quantity $c$ is then given by
\begin{equation}
 c = \sqrt{\frac{1 + \vec{a}.\vec{b}}{2}}, \; (\vec{a}.\vec{b} > 0); \; c = \sqrt{\frac{1 - \vec{a}.\vec{b}}{2}}, \; (\vec{a}.\vec{b} < 0)\; . \nonumber
\end{equation}

A similar approach has been used to obtain optimal entropic uncertainty relations for a pair of qubit observables, both in the case of the Shannon entropy~\cite{sanchez:entropyD2,  GMR03_optimalH1} and the collision entropy~\cite{BPP12_H2entropy}. Further details of our proof can be found in the appendix (Sec.~\ref{sec:qubit_EUR}).

We also show that the bound in Eq.~\eqref{eq:optimal_qubit} is tight. When $(\vec{a}.\vec{b})^{2} =1$, $c=0$; $A$ and $B$ commute and the RHS of \eqref{eq:optimal_qubit} reduces to $0$. This lower bound is attained for the common eigenstates of $A,B$. When $\vec{a}.\vec{b} = 0$, $c = \frac{1}{\sqrt{2}}$; $A$ and $B$ are mutually unbiased. The bound in \eqref{eq:optimal_qubit} is $\frac{1}{4}$, which is attained for any eigenstate of $A$ or $B$. For any other value of $\vec{a}.\vec{b}$, the lower bound is attained for the states whose Bloch vectors bisect the angle between $\vec{a}$ and $\vec{b}$. The minimizing states are thus given by
\begin{equation}
 |\psi_{\pm}\rangle\langle\psi_{\pm}| =  \frac{1}{2}\left(\Id + \left[\frac{\vec{a}\pm\vec{b}}{2c}\right].\vec{\sigma}\right).
\end{equation}

Since our disturbance measure $\cD_{F}(A;|\psi\rangle)$ is in fact the same as the entropy $T_{2}(A;|\psi\rangle)$, the tradeoff relation in Eq.~\eqref{eq:optimal_qubit} is nothing but a tight entropic uncertainty relation for the $T_{2}$ entropy:
\begin{equation}
\frac{1}{2}\left[T_{2}(A;|\psi\rangle) + T_{2}(B;|\psi\rangle)\right] \geq \frac{1}{2}(1-c^{2}),
\end{equation}
Our result for $T_{2}$ assumes importance in the light of the fact that such optimal analytical bounds are known only for a handful of entropic functions, namely, the R\'enyi entropies $H_{2}$~\cite{BPP12_H2entropy}, $H_{1/2}$, and the Tsallis entropy $T_{1/2}$~\cite{rastegin2012_T1/2}. For the Shannon entropy, there is in general only a numerical estimate of the bound~\cite{sanchez:entropyD2,GMR03_optimalH1}.

\section{Disturbance tradeoff and uncertainty tradeoff for non-projective measurements}\label{sec:dist_POVM}

We now consider the general discrete observables characterized by POV measures and associated measurement transformations characterized by CP instruments. For this general class of measurements, we show that while the uncertainty and disturbance tradeoffs are both operative, they significantly differ from each other. We first note the following necessary and sufficient condition for zero uncertainty tradeoff for a pair of POVMs. The well known Shannon EUR lower bound for POVMs~\cite{KP02} is consistent with this condition.

\begin{lemma}\label{lem:unc_POVM}
Let $\cA\sim\{A_{i}\}$ and $\cB\sim\{B_{i}\}$ be two POVMs. If $S$ is any suitable entropy measure, $S(\cA;|\psi\rangle) = S(\{\langle\psi|A_{i}|\psi\rangle\})$ and $S(\cB;|\psi\rangle) = S(\{\langle\psi|B_{i}|\psi\rangle\})$ are the corresponding uncertainties of $\cA$ and $\cB$ in state $|\psi\rangle$. Then, the lower bound on the average uncertainty,
\[ c_{S}(\cA,\cB) = \inf_{|\psi\rangle} \frac{1}{2}\left[S(\cA;|\psi\rangle) + S(\cB;|\psi\rangle) \right], \]
vanishes iff there exists a state $|\psi\rangle$ such that $A_{k}|\psi\rangle = B_{l}|\psi\rangle = |\psi\rangle$, for some $k,l$.
\end{lemma}
\begin{proof}
By the standard property of any entropic measure $S$, $c_{S}(A,B) = 0$ iff $\langle\psi|A_{i}|\psi\rangle = \delta_{ik}$ and $\langle\psi|B_{j}|\psi\rangle = \delta_{jl}$ for some $k,l$. The above result then follows from the property of POVMs that $0\leq A_{i}, B_{j} \leq \Id$ and $\sum_{i}A_{i} = \sum_{j}B_{j} = \Id$.
\end{proof}

While the uncertainty tradeoff for a pair of POVM observables depends only on the positive operators characterizing the observables, the associated disturbance tradeoff crucially depends on the CP instruments which implement the measurements of these observables. We now state the basic result concerning the disturbance tradeoff principle for such observables.

\begin{theorem}\label{thm:avgdist_POVM}
For a pair of discrete POVMs $\cA\sim\{A_{i}\}$ and $\cB\sim\{B_{j}\}$, whose measurements are implemented by appropriate CP instruments $\Phi^{\cA}$ and $\Phi^{\cB}$, the average disturbance
\begin{equation}
d_{\alpha}(\cA,\cB) \equiv \inf_{|\psi\rangle} \frac{1}{2}\left[\mathfrak{D}_{\alpha}(\cA;|\psi\rangle) + \mathfrak{D}_{\alpha}(\cB;|\psi\rangle) \right], \label{eq:inf_disturbance}
\end{equation}
satisfies $0\leq d_{\alpha}(\cA,\cB) \leq 1 \; (\alpha \in \{1,F,\infty\})$. Further, if the measurements of $\cA$ and $\cB$ are implemented by L\"uders channels $\Phi_{\cL}^{\cA}$ and $\Phi_{\cL}^{\cB}$, then,
\begin{itemize}
\item [(i)]  $d_{\alpha}(\cA,\cB) = 0$ if and only if all the positive operators $\{A_{1}, A_{2}, \ldots, B_{1}, B_{2}, \ldots\}$ have a common eigenvector.
\item[(ii)] For any suitable entropy measure $S$, $c_{S}(\cA,\cB) = 0 \Rightarrow d_{\alpha}(\cA,\cB) = 0$, but not vice-versa.
\end{itemize}
\end{theorem}
\begin{proof}
The fact that $0 \leq d_{\alpha}(\cA,\cB) \leq 1$ follows from the fact that $0 \leq \mathfrak{D}_{\alpha}(X;|\psi\rangle) \leq 1$ for $X=\cA,\cB$, for $\alpha \in \{1,F,\infty\}$. For POVMs realized by L\"uders channels, we know from Lemma~\ref{lem:zero_dist} that their individual disturbances vanish only in a common eigenstate of the POV elements. Hence, $d_{\alpha}(\cA,\cB) = 0$ if and only if  there exists a common eigenvector of {\it all} the positive operators $\{A_{1}, A_{2}, \ldots, B_{1}, B_{2}, \ldots\}$.

From Lemma~\ref{lem:unc_POVM} we know that $c_{S}(\cA,\cB) = 0$ iff there exists a state $|\psi\rangle$ such that $A_{k}|\psi\rangle = B_{l}|\psi\rangle = |\psi\rangle$, for some $k,l$. From the definition of POVMs this implies that $A_{i}|\psi\rangle = \delta_{ik}|\psi\rangle$ and $B_{j}|\psi\rangle = \delta_{jl}|\psi\rangle$. In other words, $|\psi\rangle$ is a common eigenvector of all the positive operators $\{A_{1}, A_{2}, \ldots, B_{1}, B_{2}, \ldots\}$. Thus for a pair of L\"uders instruments we see that $c_{S}(\cA,\cB) = 0$ implies that $d_{S}(\cA,\cB) = 0$.

However, $d_{\alpha}(\cA,\cB)$ can vanish even when the uncertainty tradeoff is strictly positive. We demonstrate this with an example in the appendix (Sec.~\ref{sec:dist_uncPOVM}).
\end{proof}

We have shown above that the disturbance tradeoff for a pair of L\"uders POVMs vanishes whenever the corresponding entropic uncertainty tradeoff vanishes, but not vice-versa. We can illustrate this further by comparing the fidelity-based disturbance measure $\mathfrak{D}_{F}(\cA;|\psi\rangle)$ for a POVM $\cA$ implemented by a L\"uders channel $\Phi^{\cA}$ measured on state $|\psi\rangle$, with the $T_{2}$ entropy of the corresponding probability distribution. Unlike in the case of projective measurements, the fidelity-based disturbance is less than or equal to the Tsallis $T_{2}$ entropy.
\begin{eqnarray}
\mathfrak{D}_{F}(\cA;|\psi\rangle) &=& 1 - F^{2} [\sum_{i}A_{i}^{1/2}|\psi\rangle\langle\psi|A_{i}^{1/2}, |\psi\rangle\langle\psi|] \nonumber \\
&=& 1 - \sum_{i}\langle \psi|A_{i}^{1/2}|\psi\rangle^{2} \nonumber \\
&\leq& 1 - \sum_{i}\langle \psi|A_{i}|\psi\rangle^{2} = T_{2}(\cA;|\psi\rangle).%
\end{eqnarray}
Thus for any state $|\psi\rangle$, $T_{2}(\cA;|\psi\rangle) = 0$ implies $\mathfrak{D}_{F}(\cA;|\psi\rangle) = 0$.

Furthermore, using non-projective measurements which are more general than the L\"uders class, we can construct examples where the uncertainty tradeoff $c_{S}(\cA, \cB)$ vanishes for a particular state $|\psi\rangle$, but the disturbances $\mathfrak{D}_{\alpha}(\cA; |\psi\rangle)$ and $\mathfrak{D}(\cB ;|\psi\rangle)$ are both strictly positive.

Thus we see that for non-projective measurements, the disturbance and uncertainty measures behave very differently. Hence, the disturbance tradeoff and the uncertainty tradeoff are two independent principles which reflect different aspects of incompatibility of quantum measurements.


\section{Concluding Remarks}

To summarize, we demonstrate a fundamental principle of disturbance tradeoff for incompatible quantum measurements. The existence of such a tradeoff principle implies that quantum theory places a fundamental restriction on the disturbances associated with a set of incompatible measurements, even when they are performed on distinct, identically prepared copies of a pure state.

Though disturbance and uncertainty are operationally very different concepts, for the class of projective measurements we prove a mathematical equivalence between the fidelity-based disturbance measure and the $T_{2}$ Tsallis entropy for pure states. Our work thus provides a new operational significance to the Tsallis entropy in the context of quantum foundations and quantum information theory. We also demonstrate a disturbance tradeoff principle for non-projective measurements and show how it is distinct from the uncertainty tradeoff for such measurements.

Our results thus bring to light an interesting aspect of incompatibility in quantum theory, namely, that over and in addition to the tradeoff in uncertainties, there is a tradeoff in the measurement-induced-disturbances also.




\begin{thebibliography}{32}
\expandafter\ifx\csname natexlab\endcsname\relax\def\natexlab#1{#1}\fi
\expandafter\ifx\csname bibnamefont\endcsname\relax
  \def\bibnamefont#1{#1}\fi
\expandafter\ifx\csname bibfnamefont\endcsname\relax
  \def\bibfnamefont#1{#1}\fi
\expandafter\ifx\csname citenamefont\endcsname\relax
  \def\citenamefont#1{#1}\fi
\expandafter\ifx\csname url\endcsname\relax
  \def\url#1{\texttt{#1}}\fi
\expandafter\ifx\csname urlprefix\endcsname\relax\def\urlprefix{URL }\fi
\providecommand{\bibinfo}[2]{#2}
\providecommand{\eprint}[2][]{\url{#2}}

\bibitem[{\citenamefont{Heisenberg}(1927)}]{heisenberg:ur}
\bibinfo{author}{\bibfnamefont{W.}~\bibnamefont{Heisenberg}},
  \bibinfo{journal}{Zeitschrift f{\"u}r Physik} \textbf{\bibinfo{volume}{43}},
  \bibinfo{pages}{172} (\bibinfo{year}{1927}).

\bibitem[{\citenamefont{Robertson}(1929)}]{robertson:uncertainty}
\bibinfo{author}{\bibfnamefont{H.}~\bibnamefont{Robertson}},
  \bibinfo{journal}{Physical Review} \textbf{\bibinfo{volume}{34}},
  \bibinfo{pages}{163} (\bibinfo{year}{1929}).

\bibitem[{\citenamefont{Schr\"odinger}(1930)}]{schrodinger}
\bibinfo{author}{\bibfnamefont{E.}~\bibnamefont{Schr\"odinger}},
  \bibinfo{journal}{Sitzungsberichte der Preussischen Akademie der
  Wissenschaften, Physikalisch mathematische Klasse}
  \textbf{\bibinfo{volume}{14}}, \bibinfo{pages}{296} (\bibinfo{year}{1930}).

\bibitem[{\citenamefont{Busch et~al.}(2013)\citenamefont{Busch, Lahti, and
  Werner}}]{BLW13_Heisenberg}
\bibinfo{author}{\bibfnamefont{P.}~\bibnamefont{Busch}},
  \bibinfo{author}{\bibfnamefont{P.}~\bibnamefont{Lahti}}, \bibnamefont{and}
  \bibinfo{author}{\bibfnamefont{R.~F.} \bibnamefont{Werner}},
  \bibinfo{journal}{Physical Review Letters} \textbf{\bibinfo{volume}{111}},
  \bibinfo{pages}{160405} (\bibinfo{year}{2013}).

\bibitem[{\citenamefont{Wehner and Winter}(2010)}]{ww:urSurvey}
\bibinfo{author}{\bibfnamefont{S.}~\bibnamefont{Wehner}} \bibnamefont{and}
  \bibinfo{author}{\bibfnamefont{A.}~\bibnamefont{Winter}},
  \bibinfo{journal}{New Journal of Physics} \textbf{\bibinfo{volume}{12}},
  \bibinfo{pages}{025009} (\bibinfo{year}{2010}).

\bibitem[{\citenamefont{Bialynicki-Birula and
  Mycielski}(1975)}]{bm:uncertainty}
\bibinfo{author}{\bibfnamefont{I.}~\bibnamefont{Bialynicki-Birula}}
  \bibnamefont{and}
  \bibinfo{author}{\bibfnamefont{J.}~\bibnamefont{Mycielski}},
  \bibinfo{journal}{Communications in Mathematical Physics}
  \textbf{\bibinfo{volume}{44}}, \bibinfo{pages}{129} (\bibinfo{year}{1975}).

\bibitem[{\citenamefont{Deutsch}(1983)}]{deutsch:ur}
\bibinfo{author}{\bibfnamefont{D.}~\bibnamefont{Deutsch}},
  \bibinfo{journal}{Physical Review Letters} \textbf{\bibinfo{volume}{50}},
  \bibinfo{pages}{631} (\bibinfo{year}{1983}).

\bibitem[{\citenamefont{Kraus}(1987)}]{kraus:ur}
\bibinfo{author}{\bibfnamefont{K.}~\bibnamefont{Kraus}},
  \bibinfo{journal}{Physical Review D} \textbf{\bibinfo{volume}{35}},
  \bibinfo{pages}{3070} (\bibinfo{year}{1987}).

\bibitem[{\citenamefont{Maassen and Uffink}(1988)}]{maassen:ur}
\bibinfo{author}{\bibfnamefont{H.}~\bibnamefont{Maassen}} \bibnamefont{and}
  \bibinfo{author}{\bibfnamefont{J.}~\bibnamefont{Uffink}},
  \bibinfo{journal}{Physical Review Letters} \textbf{\bibinfo{volume}{60}},
\bibinfo{pages}{1103} (\bibinfo{year}{1988}).

\bibitem[{\citenamefont{Ivanovic}(1992)}]{ivanovic:ur}
\bibinfo{author}{\bibfnamefont{I.~D.} \bibnamefont{Ivanovic}},
  \bibinfo{journal}{J. Phys. A: Math. Gen.} \textbf{\bibinfo{volume}{25}},
  \bibinfo{pages}{363} (\bibinfo{year}{1992}).

\bibitem[{\citenamefont{Sanchez}(1993)}]{sanchez:old}
\bibinfo{author}{\bibfnamefont{J.}~\bibnamefont{Sanchez}},
  \bibinfo{journal}{Physics Letters A} \textbf{\bibinfo{volume}{173}},
  \bibinfo{pages}{233} (\bibinfo{year}{1993}).

\bibitem[{\citenamefont{Sanchez-Ruiz}(1998)}]{sanchez:entropyD2}
\bibinfo{author}{\bibfnamefont{J.}~\bibnamefont{Sanchez-Ruiz}},
  \bibinfo{journal}{Physics Letters A} \textbf{\bibinfo{volume}{244}},
  \bibinfo{pages}{189} (\bibinfo{year}{1998}).

\bibitem[{\citenamefont{Bialynicki-Birula}(2006)}]{bia:uncertaintyRenyi}
\bibinfo{author}{\bibfnamefont{I.}~\bibnamefont{Bialynicki-Birula}},
  \bibinfo{journal}{Physical Review A} \textbf{\bibinfo{volume}{74}},
  \bibinfo{pages}{052101} (\bibinfo{year}{2006}).

\bibitem[{\citenamefont{Wehner and Winter}(2008)}]{ww:cliffordUR}
\bibinfo{author}{\bibfnamefont{S.}~\bibnamefont{Wehner}} \bibnamefont{and}
  \bibinfo{author}{\bibfnamefont{A.}~\bibnamefont{Winter}},
  \bibinfo{journal}{Journal of Mathematical Physics}
  \textbf{\bibinfo{volume}{49}}, \bibinfo{pages}{062105}
  (\bibinfo{year}{2008}).

\bibitem[{\citenamefont{Tsallis}(1988)}]{tsallis:entropy}
\bibinfo{author}{\bibfnamefont{C.}~\bibnamefont{Tsallis}}, \bibinfo{journal}{J.
  Stat. Phys.} \textbf{\bibinfo{volume}{51}}, \bibinfo{pages}{479}
  (\bibinfo{year}{1988}).

\bibitem[{\citenamefont{Rajagopal}(1995)}]{tsallis1}
\bibinfo{author}{\bibfnamefont{A.~K.} \bibnamefont{Rajagopal}},
  \bibinfo{journal}{Physics Letters A} \textbf{\bibinfo{volume}{205}},
  \bibinfo{pages}{32} (\bibinfo{year}{1995}).

\bibitem[{\citenamefont{Rastegin}(2013)}]{rastegin2013_tsallis}
\bibinfo{author}{\bibfnamefont{A.~E.} \bibnamefont{Rastegin}},
  \bibinfo{journal}{Quantum information processing}
  \textbf{\bibinfo{volume}{12}}, \bibinfo{pages}{2947} (\bibinfo{year}{2013}).

\bibitem[{\citenamefont{Nielsen and Chuang}(2000)}]{NCbook}
\bibinfo{author}{\bibfnamefont{M.~A.} \bibnamefont{Nielsen}} \bibnamefont{and}
  \bibinfo{author}{\bibfnamefont{I.~L.} \bibnamefont{Chuang}},
  \emph{\bibinfo{title}{Quantum Computation and Quantum Information}}
  (\bibinfo{publisher}{Cambridge University Press}, \bibinfo{year}{2000}).

\bibitem[{\citenamefont{Mandayam and Srinivas}(2014)}]{PM_MDS14}
\bibinfo{author}{\bibfnamefont{P.}~\bibnamefont{Mandayam}} \bibnamefont{and}
  \bibinfo{author}{\bibfnamefont{M.~D.} \bibnamefont{Srinivas}},
  \bibinfo{journal}{arXiv preprint quant-ph:1402.0810}  (\bibinfo{year}{2014}).

\bibitem[{\citenamefont{Srinivas}(2001)}]{MDS01}
\bibinfo{author}{\bibfnamefont{M.~D.} \bibnamefont{Srinivas}},
  \emph{\bibinfo{title}{Measurements and quantum probabilities}}
  (\bibinfo{publisher}{Universities Press}, \bibinfo{year}{2001}).

\bibitem[{\citenamefont{Srinivas}(2003)}]{MDS03_successiveEUR}
\bibinfo{author}{\bibfnamefont{M.~D.} \bibnamefont{Srinivas}},
  \bibinfo{journal}{Pramana} \textbf{\bibinfo{volume}{60}},
  \bibinfo{pages}{1137} (\bibinfo{year}{2003}).

\bibitem[{\citenamefont{Ozawa}(2004)}]{ozawa2004}
\bibinfo{author}{\bibfnamefont{M.}~\bibnamefont{Ozawa}},
  \bibinfo{journal}{Annals of Physics} \textbf{\bibinfo{volume}{311}},
  \bibinfo{pages}{350} (\bibinfo{year}{2004}).

\bibitem[{\citenamefont{Buscemi et~al.}(2014)\citenamefont{Buscemi, Hall,
  Ozawa, and Wilde}}]{BHOW_PRL14}
\bibinfo{author}{\bibfnamefont{F.}~\bibnamefont{Buscemi}},
  \bibinfo{author}{\bibfnamefont{M.~J.~W.} \bibnamefont{Hall}},
  \bibinfo{author}{\bibfnamefont{M.}~\bibnamefont{Ozawa}}, \bibnamefont{and}
  \bibinfo{author}{\bibfnamefont{M.~M.} \bibnamefont{Wilde}},
  \bibinfo{journal}{Phys. Rev. Lett.} \textbf{\bibinfo{volume}{112}},
  \bibinfo{pages}{050401} (\bibinfo{year}{2014}).

\bibitem[{\citenamefont{Coles and Furrer}(2013)}]{ColesFurrer_13}
\bibinfo{author}{\bibfnamefont{P.~J.} \bibnamefont{Coles}} \bibnamefont{and}
  \bibinfo{author}{\bibfnamefont{F.}~\bibnamefont{Furrer}},
  \bibinfo{journal}{arXiv preprint quant-ph:1311.7637}  (\bibinfo{year}{2013}).


\bibitem[{\citenamefont{Heinosaari_book}(2011)}]{Heinosaari_book}
\bibinfo{author}{\bibfnamefont{T.}~\bibnamefont{Heinosaari}} \bibnamefont{and}
 \bibinfo{author}{\bibfnamefont{M.}~\bibnamefont{Ziman}},
\emph{\bibinfo{title}{The mathematical language of quantum theory: From uncertainty to entanglement}} (\bibinfo{publisher}{Cambridge University Press, 2011}).

\bibitem[{\citenamefont{busch_singh98}(1998)\citenamefont{Busch and Singh}}]{busch_singh98}
\bibinfo{author}{\bibfnamefont{P.}~\bibnamefont{Busch}}
\bibnamefont{and}
\bibinfo{author}{\bibfnamefont{J.}~\bibnamefont{Singh}}, \bibinfo{journal}{Physics Letters A} \textbf{\bibinfo{volume}{249}}, \bibinfo{pages}{10}  (\bibinfo{year}{1998}).

\bibitem[{\citenamefont{AGG02}(2002)\citenamefont{Alvaro et al}}]{AGG02}
\bibinfo{author}{\bibnamefont{A.}~\bibnamefont{Alvaro}},
\bibinfo{author}{\bibnamefont{A.}~\bibnamefont{Gheondea}} \bibnamefont{and}  \bibinfo{author}{\bibnamefont{S.}~\bibnamefont{Gudder}},
\bibinfo{journal}{Journal of Mathematical Physics} \textbf{\bibinfo{volume}{43}}, \bibinfo{pages}{ 5872-5881} (\bibinfo{year}{2002}).

\bibitem[{\citenamefont{Bengtsson and {\.Z}yczkowski}(2006)}]{bengtsson_2006}
\bibinfo{author}{\bibfnamefont{I.}~\bibnamefont{Bengtsson}} \bibnamefont{and}
  \bibinfo{author}{\bibfnamefont{K.}~\bibnamefont{{\.Z}yczkowski}},
  \emph{\bibinfo{title}{Geometry of quantum states: an introduction to quantum
  entanglement}} (\bibinfo{publisher}{Cambridge University Press},
  \bibinfo{year}{2006}).

\bibitem[{\citenamefont{Majern{\'\i}k and
  Majern{\'\i}kov{\'a}}(2001)}]{majernik_2001}
\bibinfo{author}{\bibfnamefont{V.}~\bibnamefont{Majern{\'\i}k}}
  \bibnamefont{and}
  \bibinfo{author}{\bibfnamefont{E.}~\bibnamefont{Majern{\'\i}kov{\'a}}},
  \bibinfo{journal}{Reports on Mathematical Physics}
  \textbf{\bibinfo{volume}{47}}, \bibinfo{pages}{381} (\bibinfo{year}{2001}).

\bibitem[{\citenamefont{Niekamp et~al.}(2012)\citenamefont{Niekamp, Kleinmann,
  and G{\"u}hne}}]{NKG12_stabilizerEUR}
\bibinfo{author}{\bibfnamefont{S.}~\bibnamefont{Niekamp}},
  \bibinfo{author}{\bibfnamefont{M.}~\bibnamefont{Kleinmann}},
  \bibnamefont{and}
  \bibinfo{author}{\bibfnamefont{O.}~\bibnamefont{G{\"u}hne}},
  \bibinfo{journal}{Journal of Mathematical Physics}
  \textbf{\bibinfo{volume}{53}}, \bibinfo{pages}{012202}
  (\bibinfo{year}{2012}).

\bibitem[{\citenamefont{Rastegin}(2012)}]{rastegin2012_T1/2}
\bibinfo{author}{\bibfnamefont{A.~E.} \bibnamefont{Rastegin}},
  \bibinfo{journal}{International Journal of Theoretical Physics}
  \textbf{\bibinfo{volume}{51}}, \bibinfo{pages}{1300} (\bibinfo{year}{2012}).

\bibitem[{\citenamefont{Wu et~al.}(2009)\citenamefont{Wu, Yu, and
  Molmer}}]{wu:MUB_bound}
\bibinfo{author}{\bibfnamefont{S.}~\bibnamefont{Wu}},
  \bibinfo{author}{\bibfnamefont{S.}~\bibnamefont{Yu}}, \bibnamefont{and}
  \bibinfo{author}{\bibfnamefont{K.}~\bibnamefont{Molmer}},
  \bibinfo{journal}{Physical Review A} \textbf{\bibinfo{volume}{79}},
  \bibinfo{pages}{022104} (\bibinfo{year}{2009}).

\bibitem[{\citenamefont{Larsen}(1990)}]{larsen:entropy}
\bibinfo{author}{\bibfnamefont{U.}~\bibnamefont{Larsen}}, \bibinfo{journal}{J.
  Phys. A: Math. Gen.} \textbf{\bibinfo{volume}{23}}, \bibinfo{pages}{1041}
  (\bibinfo{year}{1990}).

\bibitem[{\citenamefont{Ghirardi et~al.}(2003)\citenamefont{Ghirardi,
  Marinatto, and Romano}}]{GMR03_optimalH1}
\bibinfo{author}{\bibfnamefont{G.}~\bibnamefont{Ghirardi}},
  \bibinfo{author}{\bibfnamefont{L.}~\bibnamefont{Marinatto}},
  \bibnamefont{and} \bibinfo{author}{\bibfnamefont{R.}~\bibnamefont{Romano}},
  \bibinfo{journal}{Physics Letters A} \textbf{\bibinfo{volume}{317}},
  \bibinfo{pages}{32} (\bibinfo{year}{2003}).

\bibitem[{\citenamefont{Bosyk et~al.}(2012)\citenamefont{Bosyk, Portesi, and
  Plastino}}]{BPP12_H2entropy}
\bibinfo{author}{\bibfnamefont{G.~M.}~\bibnamefont{Bosyk}},
  \bibinfo{author}{\bibfnamefont{M.}~\bibnamefont{Portesi}}, \bibnamefont{and}
  \bibinfo{author}{\bibfnamefont{A.}~\bibnamefont{Plastino}},
  \bibinfo{journal}{Physical Review A} \textbf{\bibinfo{volume}{85}},
  \bibinfo{pages}{012108} (\bibinfo{year}{2012}).

\bibitem[{\citenamefont{KP}(2002)\citenamefont{Krishna and
  Parthasarathy}}]{KP02}
\bibinfo{author}{\bibfnamefont{M.}~\bibnamefont{Krishna}} \bibnamefont{and}
  \bibinfo{author}{\bibfnamefont{K.~R.}~\bibnamefont{Parthasarathy}},
  \bibinfo{journal}{Sankhya: The Indian Journal of Statistics} \textbf{\bibinfo{volume}{64}},
  \bibinfo{pages}{842} (\bibinfo{year}{2002}).


\end{thebibliography}

\appendix

\section{Proof of optimal $T_{2}$ EUR for qubit observables}\label{sec:qubit_EUR}
The proof of Theorem~\ref{thm:optimal_qubit} is based on the following optimal entropic uncertainty relation (EUR) in terms of the Tsallis entropy $T_{2}$, for a pair of qubit observables. Optimal EURs for the Tsallis entropies $T_{\alpha}$ have only been obtained for $\alpha = \frac{1}{2},1$~\cite{sanchez:entropyD2, GMR03_optimalH1, rastegin2012_T1/2} thus far.
\begin{theorem}
For a pair of qubit observables $A, B$ with discrete spectra $A = \sum_{i=1}^{2}a_{i}|a_{i}\rangle\langle a_{i}|$, $B = \sum_{j=1}^{2}b_{j}|b_{j}\rangle\langle b_{j}|$, and any state $\rho$,
\begin{equation}
\inf_{\rho}\frac{1}{2}\left[T_{2}(A;\rho) + T_{2}(B;\rho)\right] \geq \frac{1}{2}(1-c^{2}), \label{eq:T2optimal}
\end{equation}
where $c \equiv \max_{i,j=1,2}|\langle a_{i}|b_{j}\rangle|$.
\end{theorem}

\begin{proof}
Since the $T_{2}$ entropy is concave, it suffices to minimize the sum of entropies over pure states. Thus, the quantity we seek to evaluate is:
\begin{eqnarray}
 c(A,B) &\equiv& \frac{1}{2}\inf_{|\psi\rangle}\left[T_{2}(A;|\psi\rangle) + T_{2}(B;|\psi\rangle)\right] \label{eq:c_qubit} \\
 &=& \inf_{|\psi\rangle}\left[1 - \frac{1}{2}\left(\sum_{i}\left(p_{|\psi\rangle}^{A}(i)\right)^{2} + \sum_{i}\left(p_{|\psi\rangle}^{B}(i)\right)^{2}\right)\right] \nonumber
\end{eqnarray}

Let $A$ and $B$ have an associated Bloch sphere representation in terms of unit vectors $\vec{a}, \vec{b} \in \mathbb{R}^{3}$. Specifically,
\begin{equation}
A = \alpha_{1}\Id + \alpha_{2}\vec{a}.\vec{\sigma} ; \; B = \beta_{1}\Id + \beta_{2}\vec{b}.\vec{\sigma}, \nonumber
\end{equation}
where, $\{\alpha_{i},\beta_{i}\}$ are real parameters and $\vec{\sigma} = (\sigma_{X}, \sigma_{Y}, \sigma_{Z})$ denote the Pauli matrices and $\Id$ denotes the $2\times 2$ identity matrix. Any pure state can similarly be denoted in terms of a unit vector $\vec{r} \in \mathbb{R}^{3}$, as $|\psi\rangle\langle\psi| = \frac{1}{2}(\Id + \vec{r}.\vec{\sigma})$.

Rewriting the probabilities in terms of the vectors $\vec{a}, \vec{b}$ and $\vec{r}$, the quantity in Eq.~\eqref{eq:c_qubit} becomes,
\begin{equation}
c_{\vec{a},\vec{b}} = \frac{1}{2}\inf_{\vec{r}}\left[ 1 - \frac{1}{2}\left((\vec{a}.\vec{r})^{2} + (\vec{b}.\vec{r})^{2}\right)\right] = \frac{1}{2}\inf_{\vec{r}}\cF_{\vec{a},\vec{b}}(\vec{r}) \nonumber
\end{equation}

Thus, the average disturbance function we seek to minimize is of the form
\begin{equation}
\cF_{\vec{a},\vec{b}}(\vec{r}) = \left[ 1 - \frac{1}{2}\left((\vec{a}.\vec{r})^{2} + (\vec{b}.\vec{r})^{2}\right)\right] \nonumber
\end{equation}
Closely following the earlier work of Ghirardi {\it et al}~\cite{GMR03_optimalH1} and Bosyk {\it et al}~\cite{BPP12_H2entropy}, we first argue that the minimizing vector $\vec{r}$ must be coplanar with $\vec{a}$ and $\vec{b}$. Let $\cP$ denote the plane determined by the vectors $\vec{a}$ and $\vec{b}$.  Given any unit vector $\vec{v}_{\perp}$ in a plane $\cP_{\perp}$ perpendicular to $\cP$, there exists a vector $v_{c}$ in the intersection of $\cP$ and $\cP_{\perp}$, such that $|\vec{a}.\vec{v}_{\perp}| \leq |\vec{a}. \vec{v}_{c}|$. Now, note that the function $f(x) = 1 - \frac{x^{2}}{2}$ is monotonically decreasing for $x\in [0,1]$. Thus, for every vector $\vec{v}_{\perp} \in \cP_{\perp}$ in a plane perpendicular to $\cP$, there exists a coplanar vector $\vec{v}_{c} \in \cP\cap\cP_{\perp}$ such that
\[ 1 -  \frac{(\vec{a}.\vec{v}_{c})^{2}}{2} \leq  1 - \frac{(\vec{a}.\vec{v}_{\perp})^{2}}{2}. \]
Making a similar argument for the vector $\vec{b}$, we see that the minimum value $\cF_{\vec{a},\vec{b}}(\vec{r})$ is attained for a vector $\vec{r} \in \cP$.

Coplanarity of $\vec{a}, \vec{b}$ and $\vec{r}$, implies that if $\theta = \cos^{-1}(\vec{a}.\vec{b})$ and $\alpha = \cos^{-1}(\vec{a}.\vec{r})$, then $\vec{b}.\vec{r} = \cos(\theta-\alpha)$. Therefore, we can rewrite the function $\cF_{\vec{a},\vec{b}}(\vec{r})$ as
\begin{equation}
\cF_{\theta}(\alpha) = 1 - \frac{1}{2}\left[\cos^{2}\alpha + \cos^{2}(\theta - \alpha) \right],
\end{equation}
reducing the problem to a minimization over a single variable $\alpha$. Since $\cF_{\theta}(\alpha)$ is periodic in $\alpha$, it suffices to minimize over the interval $\alpha \in [0,\pi]$. Differentiating with respect to $\alpha$ and setting the first derivative to zero, we see that the extremizing values of $\alpha$ satisfy
\begin{eqnarray}
 \sin(2\alpha_{*}) &=& \sin2(\theta-\alpha_{*}) \nonumber \\
 \Rightarrow \alpha_{*} &=& \frac{\theta}{2} + k\frac{\pi}{2}, \qquad k = 0,1,2,\ldots
 \end{eqnarray}
Thus, $\alpha_{+} = \frac{\theta}{2}$ and $\alpha_{-} = \frac{\theta}{2} + \frac{\pi}{2}$ are the two relevant solutions in the interval $\alpha \in [0,\pi]$. By explicitly evaluating the function $\cF_{\theta}(\alpha)$ at $\alpha_{\pm}$, we can check that $\alpha_{+}$ indeed corresponds to a minimum for $0\leq \theta \leq \frac{\pi}{2}$ and $\alpha_{-}$ corresponds to a minimum for $\frac{\pi}{2}\leq \theta \leq \pi$.

The minimum value of the average disturbance function is therefore,
\begin{equation}
 c_{\theta} = \frac{1}{2}\inf_{\alpha}\cF_{\theta}(\alpha) = \left\{ \begin{array}{ll} \frac{1}{2}(1 - \cos^{2}\frac{\theta}{2}) & {\rm for} \; 0\leq \theta \leq \frac{\pi}{2} \\
    \frac{1}{2}(1 - \sin^{2}\frac{\theta}{2}) &  {\rm for} \; \frac{\pi}{2} \leq \theta \leq \pi
     \end{array} \right\} \label{eq:c_theta}
\end{equation}

To realize the minimum value in terms of the overlap between the eigenstates of $A$ and $B$, recall that $A = \alpha_{1}\Id + \alpha_{2}\vec{a}.\vec{\sigma} ; \; B = \beta_{1}\Id + \beta_{2}\vec{a}.\vec{\sigma}$. Then, 
\[  |\langle a_{i}|b_{j}\rangle |^{2} = \left\{ \begin{array}{ll} \frac{1 + \vec{a}.\vec{b}}{2} & {\rm for} \; i=j \\
    \frac{1 - \vec{a}.\vec{b}}{2}       &  {\rm for} \; i\neq j
     \end{array} \right\} \]
Therefore, in terms of the angle $\theta$ between $\vec{a}$ and $\vec{b}$,
\begin{equation}
 |\langle a_{i}|b_{j}\rangle | = \left\{ \begin{array}{ll} \cos\frac{\theta}{2} & {\rm for} \; i=j \\
    \sin\frac{\theta}{2}       &  {\rm for} \; i\neq j
     \end{array} \right\}
 \end{equation}
In particular, defining $c \equiv \max_{i,j}|\langle a_{i}|b_{j}\rangle|$, we see that,
\begin{equation}
c = \max_{i,j}|\langle a_{i}|b_{j}\rangle| = \left\{ \begin{array}{ll} \cos\frac{\theta}{2} & {\rm for} \; 0 \leq \theta \leq \frac{\pi}{2} \\
    \sin\frac{\theta}{2} &  {\rm for} \; \frac{\pi}{2} \leq \theta \leq  \pi
       \end{array} \right\} \label{eq:c_defn}
\end{equation}
Putting together equations Eq.~\eqref{eq:c_theta} and Eq.~\eqref{eq:c_defn}, we see that
\begin{equation}
\frac{1}{2}\left[T_{2}(A;|\psi\rangle) + T_{2}(B;|\psi\rangle)\right] \geq c_{\theta} = \frac{1}{2}(1 - c^{2}).
\end{equation}

The minimizing Bloch vector $\vec{r}_{+}$ corresponding to the solution $\alpha_{+} = \frac{\theta}{2}$ satisfies
\[\vec{a}.\vec{r}_{+} = \vec{b}.\vec{r}_{+} = \cos\frac{\theta}{2}.\]
Thus, $\vec{r}_{+} = \vec{a} + \vec{b}$ up to normalization, and is the vector that bisects the interior angle between $\vec{a}$ and $\vec{b}$. The vector $\vec{r}_{-}$ corresponding to the other solution $\alpha_{-} = \frac{\theta}{2} + \frac{\pi}{2}$ satisfies
\[ \vec{a}.\vec{r}_{-} = \cos(\frac{\theta}{2} + \frac{\pi}{2}), \; \vec{b}.\vec{r}_{-} = \cos(\frac{\theta}{2} - \frac{\pi}{2}).\] Therefore, $\vec{r}_{-} = \vec{a} - \vec{b}$, upto normalization; it is the exterior angle bisector and is perpendicular to $\vec{r}_{+}$. Both minimizing vectors satisfy $|\vec{r}_{+}| = |\vec{r}_{-}| = 2c$. In summary, the Bloch vectors corresponding to the minimizing states are:
\begin{equation}
\vec{r}_{\pm} = \left\{ \begin{array}{ll} \frac{\vec{b} + \vec{a}}{2c} & {\rm for} \; 0 \leq \theta \leq \frac{\pi}{2} \\
 \frac{\vec{b} - \vec{a}}{2c} & {\rm for} \; \frac{\pi}{2} \leq \theta \leq \pi
 \end{array} \right\} \nonumber
\end{equation}
\end{proof}

\section{Proof of Theorem~\ref{thm:anticommuting}}\label{sec:proof}
Here we seek to prove the disturbance tradeoff relation in Eq.~\eqref{eq:anticommuting_dist}for a set of pairwise anticommuting observables $\{A_{1}, A_{2}, \ldots, A_{N}\}$ with eigenvalues $\{\pm 1\}$:
\begin{equation}
 \{A_{j}, A_{k}\} = 0\; \forall j\neq k, \; (A_{i})^{2} = \Id \; \forall i=1,\ldots,N . \nonumber
\end{equation}
\begin{proof}
For any pure state $|\psi\rangle$, the expectation values of such a set of $N$ dichotomic anticommuting observables are known to satisfy the following relation~\cite{ww:cliffordUR, NKG12_stabilizerEUR}:
\begin{equation}
 \sum_{i=1}^{N}(\langle \psi |A_{i}|\psi\rangle)^{2} \leq 1 . \label{eq:metaUR}
\end{equation}
Corresponding to a measurement of observable $A_{i}$ in state $|\psi\rangle$, the probabilities of obtaining outcomes $\pm 1$ are related to the expectation value $\langle \psi |A_{i}|\psi\rangle$, as follows:
\begin{equation}
p_{|\psi\rangle}^{A_{i}}(\pm 1) = \langle \psi|P_{i}^{\pm}|\psi\rangle = \frac{1 \pm \langle \psi|A_{i}|\psi\rangle}{2} .
\end{equation}
Thus the relation Eq.~\eqref{eq:metaUR} above implies the following upper bound on the sums of the probabilities:
\begin{eqnarray}
\sum_{j=1}^{N}\left[(p^{A_{i}}_{|\psi\rangle}(+))^{2} + (p^{A_{i}}_{|\psi\rangle}(-))^{2}\right] &=& \sum_{j=1}^{N}\frac{1 + (\langle\psi|A_{i}|\psi\rangle)^{2}}{2} \nonumber \\
&\leq& \frac{N + 1}{2} .
\end{eqnarray}
The disturbance tradeoff relation \eqref{eq:anticommuting_dist} follows immediately.
\end{proof}

\section{Example illustrating $ d_{\alpha}(\cA,\cB) = 0 \nRightarrow c_{S}(\cA,\cB) = 0$ }\label{sec:dist_uncPOVM}
Let $|\phi_{1}\rangle, |\phi_{2}\rangle, |\phi_{3}\rangle \in \cH_{3}$ be an orthonormal basis for a three-dimensional Hilbert space $\cH_{3}$. Consider the POVM $\cA$ described by the following positive operators:
\begin{eqnarray}
A_{1} &=& \frac{1}{6}|\phi_{1}\rangle\langle\phi_{1}| + \frac{1}{2}(|\phi_{2}\rangle\langle\phi_{2}| + |\phi_{3}\rangle\langle\phi_{3}|), \nonumber \\
A_{2} &=& \frac{2}{3}|\phi_{1}\rangle\langle\phi_{1}| + \frac{1}{2}(|\phi_{2}\rangle\langle\phi_{2}| + |\phi_{3}\rangle\langle\phi_{3}|), \nonumber \\
A_{3} &=& \frac{1}{6}|\phi_{1}\rangle\langle\phi_{1}| . \nonumber
\end{eqnarray}

Consider a different POVM $\cB$ which is constructed with vectors $\{|\phi_{1}\rangle, |\phi'_{2}\rangle, |\phi'_{3}\rangle\}$, which form another orthonormal basis for $\cH_{3}$.
\begin{eqnarray}
B_{1} &=& \frac{1}{6}|\phi_{1}\rangle\langle\phi_{1}| + \frac{1}{2}(|\phi'_{2}\rangle\langle\phi'_{2}| + |\phi'_{3}\rangle\langle\phi'_{3}|), \nonumber \\
B_{2} &=& \frac{2}{3}|\phi_{1}\rangle\langle\phi_{1}| + \frac{1}{2}(|\phi'_{2}\rangle\langle\phi'_{2}| + |\phi'_{3}\rangle\langle\phi'_{3}|), \nonumber \\
B_{3} &=& \frac{1}{6}|\phi_{1}\rangle\langle\phi_{1}| . \nonumber
\end{eqnarray}
The state $|\phi_{1}\rangle$ is a common eigenvector of all the operators $\{A_{i}, B_{j},\; i,j=1,2,3\}$. Hence measurements of $\cA, \cB$ via L\"uders instruments lead to a zero disturbance tradeoff in state $|\phi_{1}\rangle$. However, since none of the POV elements $\{A_{i}, B_{j},\; i,j=1,2,3\}$ has eigenvalue $1$, there is no state with a zero uncertainty tradeoff. $\cA,\cB$ is thus an example of a pair of POVMs which have a zero disturbance tradeoff, but at the same time, a non-zero uncertainty tradeoff. Hence, $d_{\alpha}(\cA,\cB) = 0 \nRightarrow c_{S}(\cA,\cB) = 0$.


\end{document}